\documentclass[11pt]{article}%
\usepackage{amsfonts}
\usepackage{amsmath}
\usepackage{fullpage}
\usepackage{amssymb}
\usepackage{graphicx}%
\providecommand{\U}[1]{\protect\rule{.1in}{.1in}}
\newtheorem{theorem}{Theorem}

\newtheorem{corollary}[theorem]{Corollary}

\newtheorem{lemma}[theorem]{Lemma}

\newtheorem{remark}[theorem]{Remark}

\newenvironment{proof}[1][Proof]{\noindent\textbf{#1.} }{\ \rule{0.5em}{0.5em}}
\begin{document}

\title{Monotonicity of $p$-norms of multiple operators via unitary swivels}
\author{Mark M. Wilde\thanks{Hearne Institute for Theoretical Physics, Department of
Physics and Astronomy, Center for Computation and Technology, Louisiana State
University, Baton Rouge, Louisiana 70803, USA}}
\date{July 26, 2016}
\maketitle

\begin{abstract}
Following the various statements of \cite{Dupuis2016} to their logical
conclusion, this note explicitly argues the following statement, implicit in
\cite{Dupuis2016}: for positive semi-definite operators $C_{1},\ldots,\,C_{L}%
$, a unitary $V_{C_{i}}$ commuting with $C_{i}$, and $p\geq1$, the quantity%
$$
\max_{V_{C_{1}},\ldots,V_{C_{L}}}\left\Vert C_{1}^{1/p}V_{C_{1}}\cdots
C_{L}^{1/p}V_{C_{L}}\right\Vert _{p}^{p}%
$$
is monotone non-increasing with respect to $p$. The idea from \cite{Dupuis2016} is that by allowing
unitary swivels connecting a long chain of positive semi-definite operators together, we can
establish such a statement, which might not hold generally without the presence
of the unitary swivels. Other related statements follow directly from \cite{Dupuis2016} as well, being implicit there, and are given
explicitly in this note.

\end{abstract}

\section{Introduction}

In this short note, I conduct the exercise of combining the various statements
given in \cite{Dupuis2016}\ and taking them to their logical conclusion. The
result is a monotonicity inequality regarding $p$-norms of multiple operators
strung together in a sequence. The only modification I make to the prior
statements from \cite{Dupuis2016}\ is to substitute density operators with
general positive semidefinite operators. In \cite{Dupuis2016}, my coauthor and
I were motivated by concerns in quantum information theory, and so there we
worked exclusively with density operators (positive semi-definite operators
with trace equal to one); however, it is obvious that all of the inequalities
established there extend to the more general case when the operators are
positive semi-definite with no restriction on their trace.

One of the main messages of \cite{Dupuis2016} is that it is possible to
establish non-trivial orderings of generalized R\'{e}nyi entropies formed by connecting
the marginals of density operators together in a product under a Schatten
$p$-norm, while at the same time allowing for \textquotedblleft unitary
swivels\textquotedblright\ between these operators. In \cite{Dupuis2016}, my
coauthor and I used the phrase \textquotedblleft unitary
swivels\textquotedblright\ to describe the method for arriving at the
aforementioned inequalities, because, in spite of the fact that
straightforward multi-operator extensions of the statements do not appear to
be generally true, we showed how they hold if allowing for unitary swivels
interleaved in a large chain of operators connected together. The bedrock upon
which these results rested is the powerful method of complex interpolation
\cite{BL76}, which has found a number of applications in a variety of areas in
mathematics and physics.

To begin with, let us recall the following explicit statement from
\cite[Proposition~18]{Dupuis2016}, as specialized in \cite[Corollary~19]%
{Dupuis2016}:\footnote{Here and throughout, I am following the labeling in
the arXiv post for \cite{Dupuis2016}.}%
\begin{equation}
\max_{V_{\rho_{C}}}\left\Vert \rho_{AC}^{1/p}V_{\rho_{C}}\rho_{C}^{-1/p}%
\rho_{BC}^{1/p}\right\Vert _{p}^{p}\text{ is monotone non-increasing for
}p\geq2,
\end{equation}
where $\rho_{ABC}$ is a density operator acting on a Hilbert space
$\mathcal{H}_{A}\otimes\mathcal{H}_{B}\otimes\mathcal{H}_{C}$, $\rho
_{BC}=\operatorname{Tr}_{A}\{\rho_{ABC}\}$, $\rho_{AC}=\operatorname{Tr}%
_{B}\{\rho_{ABC}\}$, and $\rho_{C}=\operatorname{Tr}_{AB}\{\rho_{ABC}\}$ are
its marginals, $V_{\rho_{C}}$ is a unitary commuting with $\rho_{C}$, and
$\left\Vert X\right\Vert _{p}\equiv\lbrack\operatorname{Tr}\{\left\vert
X\right\vert ^{p}\}]^{1/p}$ is the Schatten $p$-norm of an operator $X$. In
\cite[Section~6]{Dupuis2016}, it was discussed how one can chain together
various density operators acting on tensor-product Hilbert spaces and obtain
results similar to those given in the rest of the paper \cite{Dupuis2016}.
Carrying this through, the conclusion is that the following statement holds%
\begin{equation}
\max_{V_{C_{1}},\ldots,V_{C_{L}}}\left\Vert C_{1}^{1/p}V_{C_{1}}C_{2}%
^{1/p}V_{C_{2}}\cdots C_{L}^{1/p}V_{C_{L}}\right\Vert _{p}^{p}\text{ is
monotone non-increasing for }p\geq2,\label{eq:first-one}%
\end{equation}
for $C_{1}$, \ldots, $C_{L}$ density operators and $V_{C_{i}}$ a unitary
commuting with $C_{i}$. In \cite[Remark~12]{Dupuis2016}, it was mentioned how
the optimizations over commuting unitaries can be replaced with more explicit
bounds found by applying the Stein--Hirschman operator interpolation theorem
\cite{S56,H52}. Carrying this statement through as well, the conclusion is
that the following inequality holds for $2\leq q\leq p$:%
\begin{equation}
\log\left\Vert C_{1}^{1/p}C_{2}^{1/p}\cdots C_{L}^{1/p}\right\Vert _{p}%
^{p}\leq\int_{-\infty}^{\infty}dt\ \beta_{q/p}(t)\ \log\left\Vert
C_{1}^{\left(  1+it\right)  /q}C_{2}^{\left(  1+it\right)  /q}\cdots
C_{L}^{\left(  1+it\right)  /q}\right\Vert _{q}^{q},\label{eq:multi-op-trace}%
\end{equation}
if $C_{1}$, \ldots, $C_{L}$ are density operators and $\beta_{\theta}%
(t)\equiv\sin(\pi\theta)/(2\theta\left[  \cosh(\pi t)+\cos(\pi\theta)\right]
)$, a probability distribution over $t\in\mathbb{R}$ and with a parameter
$\theta\in\left[  0,1\right]  $. In \cite[Section~6]{Dupuis2016}, it was also
discussed how one can obtain limits of the inequalities presented in the paper
by applying the well known Lie-Trotter product formula. Carrying this through
(i.e., taking the limit $p\rightarrow\infty$), the conclusion is that the
following inequality holds%
\begin{equation}
\log\operatorname{Tr}\left\{  \exp\left\{  \log C_{1}+\cdots+\log
C_{L}\right\}  \right\}  \leq\int_{-\infty}^{\infty}dt\ \beta_{0}%
(t)\ \log\left\Vert C_{1}^{\left(  1+it\right)  /q}C_{2}^{\left(  1+it\right)
/q}\cdots C_{L}^{\left(  1+it\right)  /q}\right\Vert _{q}^{q}%
,\label{eq:multi-op-GT}%
\end{equation}
where $\beta_{0}(t)\equiv\lim_{\theta\rightarrow0}\beta_{\theta}%
(t)=\pi/\left(  2\left[  \cosh(\pi t)+1\right]  \right)  $. By inspection of
the proof given in \cite[Proposition~18]{Dupuis2016}, it is clear that the
inequalities in \eqref{eq:multi-op-trace}--\eqref{eq:multi-op-GT} hold for
positive semi-definite operators as well. We can also see from that proof that
\eqref{eq:multi-op-trace} holds more generally for $1\leq q\leq p$ and
\eqref{eq:multi-op-GT} for $1\leq q$.

\section{Explicit Proofs of \eqref{eq:first-one}--\eqref{eq:multi-op-GT}}

In the rest of this note, I give explicit proofs of
\eqref{eq:first-one}--\eqref{eq:multi-op-GT} for the benefit of the reader,
following the steps outlined in \cite{Dupuis2016} line by line.

\begin{theorem}
Let $C_{1},\ldots,C_{L}$ be positive semi-definite operators, let $V_{C_{i}}$
denote a unitary commuting with $C_{i}$ for all $i\in\left\{  1,\ldots
,L\right\}  $, and let $p\geq1$. Then the following quantity is monotone
non-increasing with respect to $p$:%
\begin{equation}
\max_{V_{C_{1}},\ldots,V_{C_{L}}}\left\Vert C_{1}^{1/p}V_{C_{1}}\cdots
C_{L}^{1/p}V_{C_{L}}\right\Vert _{p}^{p}.
\end{equation}

\end{theorem}

\begin{proof}
The proof of this statement is essentially identical to the proof of
\cite[Proposition~18]{Dupuis2016}. It is a consequence of a well known complex
interpolation theorem recalled as Lemma~\ref{thm:hadamard} below. Let
$V_{C_{1}},\ldots,V_{C_{L}}$ denote a set of fixed unitaries, where $V_{C_{i}%
}$ commutes with $C_{i}$. Let $q$ be such that $1\leq q<p$ (there is nothing
to prove if $q=p$). For $z \in \mathbb{C}$, pick%
\begin{align}
G(z) &  =C_{1}^{z/q}V_{C_{1}}\cdots C_{L}^{z/q}V_{C_{L}},\label{eq:g(z)}\\
p_{0} &  =\infty,\\
p_{1} &  =q,\\
\theta &  =q/p,\label{eq:theta}%
\end{align}
the choices above being identical to those in \cite[Eq.~(7.6)-(7.9)]%
{Dupuis2016}. This implies that $p_{\theta}=p$. Applying
Lemma~\ref{thm:hadamard} gives%
\begin{equation}
\left\Vert G(\theta)\right\Vert _{p}\leq\sup_{t\in\mathbb{R}}\left\Vert
G(it)\right\Vert _{\infty}^{1-\theta}\sup_{t\in\mathbb{R}}\left\Vert
G(1+it)\right\Vert _{q}^{\theta}.
\end{equation}
Consider that%
\begin{align}
\left\Vert G(\theta)\right\Vert _{p} &  =\left\Vert C_{1}^{1/p}V_{C_{1}}\cdots
C_{L}^{1/p}V_{C_{L}}\right\Vert _{p},\\
\left\Vert G(it)\right\Vert _{\infty} &  =\left\Vert C_{1}^{it/q}V_{C_{1}%
}\cdots C_{L}^{it/q}V_{C_{L}}\right\Vert _{\infty}\leq1,\\
\left\Vert G(1+it)\right\Vert _{q} &  =\left\Vert C_{1}^{\left(  1+it\right)
/q}V_{C_{1}}\cdots C_{L}^{\left(  1+it\right)  /q}V_{C_{L}}\right\Vert _{q}\\
&  =\left\Vert C_{1}^{1/q}C_{1}^{it/q}V_{C_{1}}\cdots C_{L}^{1/q}C_{L}%
^{it/q}V_{C_{L}}\right\Vert _{q}\\
&  \leq\max_{W_{C_{1}},\ldots,W_{C_{L}}}\left\Vert C_{1}^{1/q}W_{C_{1}}\cdots
C_{L}^{1/q}W_{C_{L}}\right\Vert _{q},
\end{align}
which are conclusions identical to those in \cite[Eq.~(7.11)-(7.17)]%
{Dupuis2016}. Putting everything together, we find that, for all $V_{C_{1}%
},\ldots,V_{C_{L}}$, the following inequality holds%
\begin{equation}
\left\Vert C_{1}^{1/p}V_{C_{1}}\cdots C_{L}^{1/p}V_{C_{L}}\right\Vert _{p}%
\leq\max_{W_{C_{1}},\ldots,W_{C_{L}}}\left\Vert C_{1}^{1/q}W_{C_{1}}\cdots
C_{L}^{1/q}W_{C_{L}}\right\Vert _{q}^{\theta},
\end{equation}
which is equivalent to%
\begin{equation}
\left\Vert C_{1}^{1/p}V_{C_{1}}\cdots C_{L}^{1/p}V_{C_{L}}\right\Vert _{p}%
^{p}\leq\max_{W_{C_{1}},\ldots,W_{C_{L}}}\left\Vert C_{1}^{1/q}W_{C_{1}}\cdots
C_{L}^{1/q}W_{C_{L}}\right\Vert _{q}^{q}.
\label{eq:almost-done}
\end{equation}
Since \eqref{eq:almost-done} holds for all $V_{C_{1}},\ldots,V_{C_{L}}$, the statement
of the theorem follows.
\end{proof}

\begin{theorem}
\label{thm:univ-ineq}Let $C_{1},\ldots,C_{L}$ be positive semi-definite
operators, and let $p>q\geq1$. Then the following inequality holds:%
\begin{equation}
\log\left\Vert C_{1}^{1/p}C_{2}^{1/p}\cdots C_{L}^{1/p}\right\Vert _{p}%
^{p}\leq\int_{-\infty}^{\infty}dt\ \beta_{q/p}(t)\ \log\left\Vert
C_{1}^{\left(  1+it\right)  /q}C_{2}^{\left(  1+it\right)  /q}\cdots
C_{L}^{\left(  1+it\right)  /q}\right\Vert _{q}^{q}.
\end{equation}

\end{theorem}

\begin{proof}
Here we directly follow the suggestion from \cite[Remark~12]{Dupuis2016}. Pick
$G(z)$, $p_{0}$, $p_{1}$, and $\theta$ as in
\eqref{eq:g(z)}--\eqref{eq:theta}, with
$V_{C_{1}}=\cdots=V_{C_{L}} = I$. Applying Lemma~\ref{thm:op-hirschman}
below, we find that%
\begin{equation}
\log\left\Vert G(\theta)\right\Vert _{p_{\theta}}\leq\int_{-\infty}^{\infty
}dt\ \alpha_{\theta}(t)\log\left\Vert G(it)\right\Vert _{p_{0}}^{1-\theta
}+\beta_{\theta}(t)\log\left\Vert G(1+it)\right\Vert _{p_{1}}^{\theta}.
\end{equation}
After using that%
\begin{equation}
\log\left\Vert G(it)\right\Vert _{p_{0}}^{1-\theta}=\log\left\Vert
C_{1}^{it/q}\cdots C_{L}^{it/q}\right\Vert _{\infty
}^{1-\theta}\leq0,
\end{equation}
as recalled above,
we are left with%
\begin{equation}
\log\left\Vert G(\theta)\right\Vert _{p_{\theta}}\leq\int_{-\infty}^{\infty
}dt\ \beta_{\theta}(t)\log\left\Vert G(1+it)\right\Vert _{p_{1}}^{\theta}.
\end{equation}
This is then equivalent to the statement of the theorem.
\end{proof}

\begin{corollary}
Let $C_{1},\ldots,C_{L}$ be positive definite operators, and let $q\geq1$.
Then the following inequality holds:%
\begin{equation}
\log\operatorname{Tr}\left\{  \exp\left\{  \log C_{1}+\cdots+\log
C_{L}\right\}  \right\}  \leq\int_{-\infty}^{\infty}dt\ \beta_{0}%
(t)\ \log\left\Vert C_{1}^{\left(  1+it\right)  /q}C_{2}^{\left(  1+it\right)
/q}\cdots C_{L}^{\left(  1+it\right)  /q}\right\Vert _{q}^{q}.
\end{equation}

\end{corollary}

\begin{proof}
Consider that
\begin{equation}
\left\Vert C_{1}^{1/2p}C_{2}^{1/2p}\cdots C_{L}^{1/2p}\right\Vert _{2p}%
^{2p}=\operatorname{Tr}\left\{  \left[  C_{L}^{1/2p}\cdots C_{2}^{1/2p}%
C_{1}^{1/p}C_{2}^{1/2p}\cdots C_{L}^{1/2p}\right]  ^{p}\right\}  .
\end{equation}
Then by the multioperator Lie--Trotter product formula \cite{S85}, we have
that%
\begin{equation}
\lim_{p\rightarrow\infty}\operatorname{Tr}\left\{  \left[  C_{L}^{1/2p}\cdots
C_{2}^{1/2p}C_{1}^{1/p}C_{2}^{1/2p}\cdots C_{L}^{1/2p}\right]  ^{p}\right\}
=\operatorname{Tr}\left\{  \exp\left\{  \log C_{1}+\cdots+\log C_{L}\right\}
\right\}  .
\end{equation}
The inequality in the statement of the corollary is then a direct consequence
of Theorem~\ref{thm:univ-ineq} and the above.
\end{proof}

\begin{lemma}
\label{thm:hadamard}Let $S\equiv\left\{  z\in\mathbb{C}:0\leq\operatorname{Re}%
\left\{  z\right\}  \leq1\right\}  $, and let $L(\mathcal{H})$ be the space of
bounded linear operators acting on a Hilbert space $\mathcal{H}$. Let
$G:S\rightarrow L(\mathcal{H})$ be a bounded map that is holomorphic on the
interior of $S$ and continuous on the boundary.\footnote{A map $G:S\rightarrow
L(\mathcal{H})$ is holomorphic (continuous, bounded) if the corresponding
functions to matrix entries are holomorphic (continuous, bounded).} Let
$\theta\in\left(  0,1\right)  $ and define $p_{\theta}$ by%
\begin{equation}
\frac{1}{p_{\theta}}=\frac{1-\theta}{p_{0}}+\frac{\theta}{p_{1}}%
,\label{eq:p-relation}%
\end{equation}
where $p_{0},p_{1}\in\lbrack1,\infty]$. For $k=0,1$ define%
\begin{equation}
M_{k}=\sup_{t\in\mathbb{R}}\left\Vert G\left(  k+it\right)  \right\Vert
_{p_{k}}.
\end{equation}
Then%
\begin{equation}
\left\Vert G\left(  \theta\right)  \right\Vert _{p_{\theta}}\leq
M_{0}^{1-\theta}M_{1}^{\theta}.\label{eq:hadamard-3-line}%
\end{equation}

\end{lemma}

The following lemma is based on Hirschman's improvement of the Hadamard
three-line theorem \cite{H52}.

\begin{lemma}[Stein--Hirschman]
\label{thm:op-hirschman} Let $S\equiv\left\{  z\in\mathbb{C}:0\leq
\operatorname{Re}\left\{  z\right\}  \leq1\right\}  $ and let $G:S\rightarrow
L(\mathcal{H})$ be a bounded map that is holomorphic on the interior of $S$
and continuous on the boundary. Let $\theta\in(0,1)$ and define $p_{\theta}$
by%
\begin{equation}
\frac{1}{p_{\theta}}=\frac{1-\theta}{p_{0}}+\frac{\theta}{p_{1}}\ ,
\end{equation}
where $p_{0},p_{1}\in\left[  1,\infty\right]  $. Then the following bound
holds%
\begin{equation}
\log\left\Vert G(\theta)\right\Vert _{p_{\theta}}\leq\int_{-\infty}^{\infty
}dt\ \alpha_{\theta}(t)\log\left\Vert G(it)\right\Vert _{p_{0}}^{1-\theta
}+\beta_{\theta}(t)\log\left\Vert G(1+it)\right\Vert _{p_{1}}^{\theta},
\label{eq:oper-hirschman}%
\end{equation}
where $\alpha_{\theta}(t)$ and $\beta_{\theta}(t)$ are defined by
\begin{align*}
\alpha_{\theta}(t)  &  \equiv\frac{\sin(\pi\theta)}{2(1-\theta)\left[
\cosh(\pi t)-\cos(\pi\theta)\right]  },\\
\beta_{\theta}(t)  &  \equiv\frac{\sin(\pi\theta)}{2\theta\left[  \cosh(\pi
t)+\cos(\pi\theta)\right]  }\ .
\end{align*}

\end{lemma}

\begin{remark}
Fix $\theta\in(0,1)$. Observe that $\alpha_{\theta}(t),\beta_{\theta}(t)\geq0$
for all $t\in\mathbb{R}$ and we have
\begin{equation}
\int_{-\infty}^{\infty}dt\ \alpha_{\theta}(t)=\int_{-\infty}^{\infty}%
dt\ \beta_{\theta}(t)=1\ ,
\end{equation}
(see, e.g., \cite[Exercise~1.3.8]{G08}) so that $\alpha_{\theta}(t)$ and
$\beta_{\theta}(t)$ can be interpreted as probability density functions.
Furthermore, the following limit holds
\begin{equation}
\lim_{\theta\searrow0}\beta_{\theta}(t)=\frac{\pi}{2\left[  \cosh(\pi
t)+1\right]  }=\beta_{0}(t)\ , \label{eq:dist-limit}%
\end{equation}
where $\beta_{0}$ is also a probability density function on $\mathbb{R}$.
\end{remark}

\bibliographystyle{alpha}
\bibliography{Ref}

\end{document}